%% file: arxiv.tex
\documentclass[aps,pra,10pt,reprint,showpacs]{revtex4-1}

\usepackage[latin1]{inputenc}

\usepackage{algpseudocode}
\usepackage{subfigure}
\usepackage{amsmath}
\usepackage{amssymb}
\usepackage{graphicx}
\usepackage{dsfont}

\newcommand{\R}{\mathbb R}
\newcommand{\1}{\mathds{1}}

\newcommand{\ket}[1]{| #1 \rangle}
\newcommand{\bra}[1]{\langle #1 |}

\usepackage{amsthm}
\usepackage{txfonts}

\newtheorem{theorem}{Theorem}
\newtheorem{corollary}{Corollary}
\newtheorem*{theorem*}{Theorem}
\theoremstyle{definition}
\newtheorem{example}{Example}
\newcommand{\tr}{\mbox{tr}\,}

\usepackage{hyperref}
\hypersetup{
    colorlinks,
    citecolor=black,
    filecolor=black,
    linkcolor=black,
    urlcolor=black
}

\begin{document}
\title{Designing Bell inequalities from a Tsirelson bound}
\author{Michael Epping}
\email{epping@thphy.uni-duesseldorf.de}
\author{Hermann Kampermann}
\author{Dagmar Bru\ss}
\affiliation{Institut f\"{u}r Theoretische Physik III, Heinrich-Heine-Universit\"{a}t D\"{u}sseldorf, Universit\"{a}tsstr. 1, D-40225
D\"{u}sseldorf, Germany}

\begin{abstract}
We present a simple analytic bound on the quantum value of general correlation type Bell inequalities, similar to Tsirelson's bound. 
It is based on the maximal singular value of the coefficient matrix associated with the inequality. We provide a criterion for tightness
of the bound and show that the class of inequalities where our bound is tight covers many famous examples from the literature. 
We describe how this bound helps to construct Bell inequalities, in particular inequalities that witness the dimension of the measured
observables.
\end{abstract}
\maketitle 
\noindent An interesting feature of quantum theory are correlations between outcomes of spatially separated measurements that contradict
predictions of all theories based on common-sense assumptions called \textit{Locality}, \textit{Reality} and \textit{Free Will}
~\cite{Bell1964,Weihs1998}. This contradiction is shown by the violation of Bell inequalities. A famous version of them was derived by
Clauser, Horne, Shimony and Holt (CHSH)~\cite{Clauser1969} and many
generalizations followed~\cite{Mermin1990,Braunstein1990,Cavalcanti2007,Vertesi2008,PhysRevA.64.014102,Gisin1999,CaslavMultiparty}. In
addition to ruling out local hidden variable theories, several other applications of Bell-type inequalities are
known~\cite{Brukner2004,PhysRevLett.67.661,Barrett2005,Pironio2010,selftestingmayers}.\\
Regarding such applications one is interested in the maximal value of the Bell expression predicted by quantum theory and the corresponding
measurements to achieve this optimum. Bounds on this quantum value were first derived by B. Tsirelson~\cite{Cirel'son1980,Tsirelson1993}. 
For general CHSH-type Bell inequalities (which will be defined later on), similar bounds can be derived. To this aim, approaches based on
different physical principles have been developed, under them
\textit{information causality}~\cite{PhysRevA.80.040103,Pawlowski,PhysRevLett.107.210403}, \textit{macroscopic
reality}~\cite{Navascues2009}, uncertainty principles~\cite{Oppenheim2010} and \textit{exclusivity}~\cite{Cabello2013}. Furthermore methods
based on semidefinite programming are known~\cite{Wehner2006,PhysRevLett.98.010401,Navascues2008,2008arXiv0803.4373D}. In contrast here we
present an analytical method to find a quantum bound, which makes use of standard
tools of linear algebra only.\\
Our bound is related to the optimization of Ref.~\cite{Wehner2006} with relaxed boundary conditions, which implies that our bound
is not necessarily reachable. However, the class of Bell inequalities reaching our bound contains most examples from the literature. We
introduce a constructive method to determine whether the bound is tight, which provides a geometric picture that allows to construct new Bell
inequalities. We exemplify this by constructing dimension witnessing Bell inequalities, analogous to the ones discussed
in~\cite{Vertesi2008,Vertesi2009,PhysRevLett.100.210503,PhysRevLett.105.230501,PhysRevLett.111.030501}. Different techniques to witness the
dimension of a quantum system are described in~\cite{PhysRevLett.102.190504,PhysRevA.78.062112}. Our construction of
new Bell
Inequalities differs from known methods based on the correlation polytope~\cite{Pitowsky1986,Peres1999} and variable
elimination~\cite{Avis2004,Budroni}.\\
We start with considering general bipartite correlation type inequalities, where the two parties $i=1,2$ measure $M_i$ different two-outcome
observables $\mathcal{A}_i(x_i)$, with $x_i=1,2,...,M_i$, on their part of the shared quantum state $\rho$ (see \cite{Peres1999} for an
overview). The principal setup of
an experiment associated with such an inequality is visualized in Fig.~\ref{fig:Setup}.%
\begin{figure}[tbp]%
 \centering%
 \includegraphics[width=\linewidth]{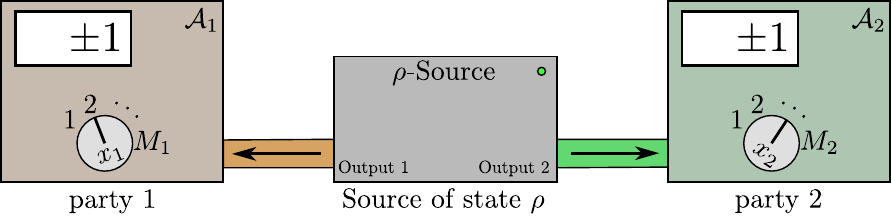}%
 \caption{\label{fig:Setup}Illustration of a bipartite Bell experiment. The source prepares the state $\rho$ and distributes one subsystem
to each party.
Each party $i$ can choose between $M_i$ different measurement settings ($\mathcal{A}_i(x_i)$, $x_i=1,...M_i$). Multiplying the two results
of both parties, which are $+1$ or $-1$, and repeating the experiment many times gives the expectation value $E(x_1,x_2)$. Bounds on linear
combinations of $E(x_1,x_2)$ for different $x_1$ and $x_2$ are discussed in the text.}%
\end{figure}%
The expectation value of the product of the measurement results of both parties in setting $x_1$ of party $1$ and setting $x_2$ of party
$2$ is denoted by $E(x_1,x_2)$. In any local and realistic theory the inequality
\begin{equation}
 \sum_{x_1=1}^{M_1}\sum_{x_2=1}^{M_2} g_{x_1,x_2} E_{\text{lr}}(x_1,x_2) \leq B, \label{eq:ineq}
\end{equation}
holds, where $g_{x_1,x_2}$ are real coefficients of a matrix $g$ and $B$ is the corresponding local hidden variable bound. 
It can be obtained by maximizing over all possible local realistic expectation values $E_{\text{lr}}(x_1,x_2)=a_1(x_1) a_2(x_2)$, where
$a_i(x_i)=\pm 1$ is the measurement result in setting $x_i$ of party $i$. 
Throughout this paper we are interested in similar bounds $T$ on the quantum value $Q$,
\begin{equation}
Q:=\max_{\rho,\mathcal{A}_1,\mathcal{A}_2} 
 \sum_{x_1=1}^{M_1}\sum_{x_2=1}^{M_2} g_{x_1,x_2} E(x_1,x_2) \leq T, \label{eq:ineqT}
\end{equation}
where $E(x_1,x_2)=\tr (\rho \mathcal{A}_1(x_1)\otimes \mathcal{A}_2(x_2))$ is the expectation value predicted by quantum theory. 
If the quantum value $Q$ violates Ineq.~(\ref{eq:ineq}), i.e. $Q>B$, we call Ineq.~(\ref{eq:ineq}) a Bell inequality.\\
We now derive an upper bound $T$ on the quantum value $Q$ using the singular value decomposition of 
the coefficient matrix $g$ (see Eq.~(\ref{eq:ineq})). For any real $M_1\times M_2$-matrix $g$ we define an orthogonal $M_1\times
M_1$-matrix $V,$ a diagonal $M_1\times M_2$-matrix $S$, containing the singular values, and an orthogonal $M_2\times M_2$-matrix $W$, such
that
\begin{equation}
 g=VSW^T. \label{eq:svdvong}
\end{equation}
We use the convention of nonincreasing order on the diagonal of $S$. The matrices $V$ and $W$ are uniquely defined up to unitary operations
on spaces associated with degenerate singular values. The maximal singular value $S_{11}$ will be written as $||g||_2$, the spectral norm of
$g$, which is defined as $||g||_2=\max_{\vec{x}, |\vec{x}|=1} |g \vec{x}|$. The multiplicity of $||g||_2$, i.e. the dimension of the 
corresponding space, is denoted by $d$. We will also use the truncated singular value decomposition associated with the maximal singular
value only. In this case the matrices are denoted $V^d$, $S^d$ and $W^d$. See Fig.~\ref{fig:matrixdimensions} for an illustration of the
dimensions of the involved matrices.\\
\begin{figure}[tbp]
 \begin{center}
  \includegraphics[width=0.6\linewidth]{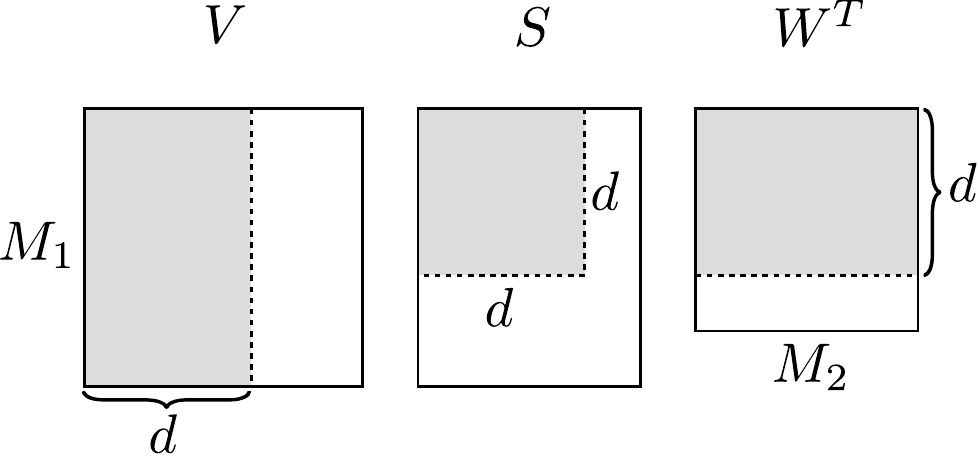}
 \end{center}
\caption{\label{fig:matrixdimensions} The dimensions of the matrices $V$, $S$ and $W$. The shaded parts belong to the truncated singular
value decomposition ($V^d$, $S^d$ and $W^d$) for the maximal singular value.}
\end{figure} 
With these definitions we can formulate the quantum bound for inequality~(\ref{eq:ineq}).
\begin{theorem}\label{thm:bound}
Let there be two parties, labeled with $i=1,2$, sharing a state given by a density matrix $\rho$, i.e. a positive semidefinite $D\times
D$-matrix, $D\in\mathds{N}$, with $\tr \rho=1$. Let $\{\mathcal{A}_i(x_i):1\leq x_i \leq M_i\}$ be a set of observables with all eigenvalues
in $[-1,1]$ on the subsystem of party $i$. The expectation value in setting $(x_1, x_2)$ is
\begin{equation}
 E(x_1,x_2)=\tr\left(\mathcal{A}_1(x_1)\otimes \mathcal{A}_2(x_2)\rho\right).
\end{equation}
For real coefficients $g_{x_1,x_2}$ the bound
 \begin{equation}
  \sum_{x_1=1}^{M_1}\sum_{x_2=1}^{M_2} g_{x_1,x_2} E(x_1,x_2)\leq \sqrt{M_1 M_2}||g||_2 =: T \label{eq:bound}
 \end{equation}
holds, where $||g||_2$ is the maximal singular value of $g$.
\end{theorem}
\begin{proof}
As the maximal value of the Bell inequality is achieved by a pure state, it is sufficient to focus on these. The basic idea is to use
a well-known result of Tsirelson~\cite{Cirel'son1980} to map physical observables to real vectors and bound the resulting expression using
their length and the maximal singular value of $g$. In order to prevent confusion, the notation of Tsirelson's theorem is adopted to the one used
here.
\begin{theorem*}[Tsirelson~\cite{Cirel'son1980}]
 Given sets of observables $\mathcal{A}_1(1),...,\mathcal{A}_1(M_1)$ and $\mathcal{A}_2(1),...,\mathcal{A}_2(M_2)$, whose eigenvalues lie
in $[-1,1]$, and an arbitrary bipartite state $\ket{\psi}\in\mathcal{H}_1\otimes\mathcal{H}_2$, there exist real unit
vectors $\vec{v}_1,...,\vec{v}_{M_1},\vec{w}_1,...,\vec{w}_{M_2}\in \R^{M_1+M_2}$ such that for all settings $x_1\in\{1,...,M_1\}$ and
$x_2\in \{1,...,M_2\}$ the expectation value can be written as
\begin{equation}
 E(x_1,x_2)=\bra{\psi} \mathcal{A}_1(x_1) \otimes \mathcal{A}_2(x_2) \ket{\psi} = \vec{v}_{x_1}^T \vec{w}_{x_2}. \label{eq:Eundvw}
\end{equation}
\end{theorem*}
\noindent This theorem ensures one can write
\begin{align}
 \sum_{x_1=1}^{M_1} \sum_{x_2=1}^{M_2} g_{x_1,x_2} E(x_1,x_2)=&\vec{V}^T (g\otimes \mathds{1}^{M_1+M_2}) \vec{W}^{\phantom{T}},
 \label{eq:bellmitvundw}
\end{align}
where we introduced the vectors 
\begin{equation}
 \vec{V}=\left(\begin{array}{c}
          \vec{v}_1\\
 \vdots\\
  \vec{v}_{M_1}
\end{array}\right)\mbox{ and }\vec{W}=\left(\begin{array}{c}
          \vec{w}_1\\
 \vdots\\
  \vec{w}_{M_2}
\end{array}\right). \label{eq:defVundW}
\end{equation}
The relation between these vectors and the matrices $V$ and $W$ from the singular value decomposition of $g$ will become clear in
Theorem~\ref{thm:tightness}. From Eq.~(\ref{eq:bellmitvundw}) we see, that $Q$ can be bounded by use of the maximal singular value of
$(g\otimes \1^{M_1+M_2})$, which is the same as the maximal singular value of $g$, and the length of $\vec{V}$ and $\vec{W}$. Because the
$\vec{v}_i$ and $\vec{w}_j$ are normalized vectors, the lengths of $\vec{V}$ and $\vec{W}$ are $\sqrt{M_1}$ and $\sqrt{M_2}$, respectively.
This finishes the proof.
\end{proof}%
\noindent The bound in Theorem~\ref{thm:bound} holds for any inequality given by an arbitrary real matrix $g$. But so far we did not
discuss the quality of the bound and indeed not for all matrices $g$ the bound is achievable (see example~\ref{ex:nottight} in
Appendix~\ref{app:examples}). In the next theorem we give a necessary and sufficient
condition for tightness of our bound.
\begin{theorem}\label{thm:tightness}
For a given real $M_1\times M_2$-matrix $g$ and the corresponding matrices $V^d$ and $W^d$ (see Fig.~\ref{fig:matrixdimensions}),
the bound~(\ref{eq:bound}) can be reached with observables, which are linked via Eq.~(\ref{eq:Eundvw}) to $d'\leq d$-dimensional real
vectors $\vec{v}_i$ and $\vec{w}_j$ given by
\begin{align}
 \vec{v}_i =& \alpha^T V^d_{i,*}, \label{eq:obsv}\\
 \vec{w}_j =& \sqrt{\frac{M_2}{M_1}} \alpha^T W^d_{j,*},\label{eq:obsw}
\end{align}
if and only if the system of equations
\begin{align}
 ||\alpha^T V^d_{i,*} ||^2=&1 &\forall i\in\{1,2,...,M_1\} \label{eq:normv}\\
 ||\alpha^T W^d_{j,*} ||^2=&\frac{M_1}{M_2}&\forall j\in\{1,2,...,M_2\}\label{eq:normw}
\end{align}
is solvable. Here the $d\times d'$-matrix $\alpha$ is the unknown and $V^d_{i,*}$ and $W^d_{j,*}$ denote column vectors
containing the elements of the $i$-th row of $V^d$ and the $j$-th row of $W^d$, respectively. 
\end{theorem}
The proof is given in Appendix~\ref{app:proofoftightness}. The main idea of the proof is, that the bound is reachable, if and only if there
exist singular vectors to the maximal singular value of $g\otimes\1^{M_1+M_2}$, $\vec{V}$ and $\vec{W}$ (see Eq.~(\ref{eq:defVundW})), where
the vectors $\vec{v}_i$ and $\vec{w}_j$ are unit vectors (see Tsirelson's theorem in Eq.~(\ref{eq:Eundvw})).\\
Note, that all vectors that
fulfill Eq.~(\ref{eq:normv}) lie on the surface of a $d$-dimensional origin-centered ellipsoid (see Fig.~\ref{fig:ellipse}). If the
vectors $V^d_{i,*}$ and $W^d_{j,*}$ permit to find an ellipsoid such that they all lie on it's surface, then the bound is tight. If semiaxes
are infinite, e.g. if the ellipsoid is not uniquely defined, then $d'<d$ and the corresponding $\alpha$ does not have full rank. In
particular $d'=1$ implies, that one-dimensional vectors reach the bound, the inequality~(\ref{eq:ineq}) cannot be violated and thus it is no
Bell inequality (see Fig.~\ref{fig:ellipseb}).
\begin{figure}[tbp]
 \subfigure[ ]{ \includegraphics[width=0.46\linewidth]{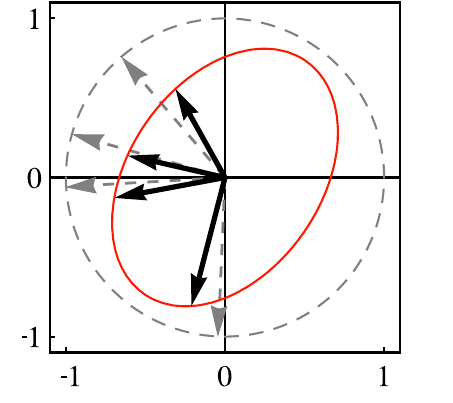} \label{fig:ellipsea}}
 \subfigure[ ]{ \includegraphics[width=0.46\linewidth]{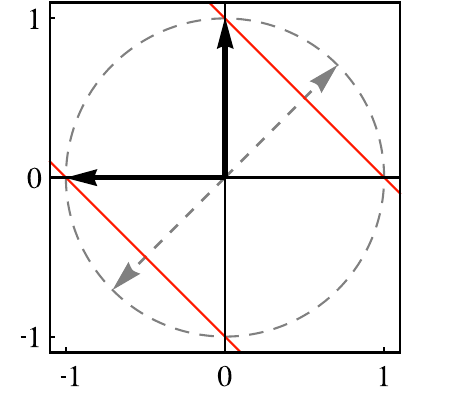} \label{fig:ellipseb}}
 \caption{\label{fig:ellipse}The vectors $V^d_{i,*}$, $i=1,2,...,M_1$ (black) can be normalized by applying the matrix $\alpha^T$, if
they lie on an origin-centered ellipsoid (red), i.e. the vectors $\vec{v}_i = \alpha^T V^d_{i,*}$ lie on the unit sphere (dashed).
An analogous picture could be drawn for $\vec{w}_j$, $j=1,2,...,M_2$. (a) In this example $d=2$ and $M_1=4$. The four vectors uniquely define an
ellipse. (b) In this example $d=2$ and $M_1=2$. This ellipse is not uniquely defined by the vectors $V^d_{i,*}$. The one shown has
one infinite semi-axis.}
\end{figure}
An algorithm solving Eqs.~(\ref{eq:normv}) and (\ref{eq:normw}) in $\mathcal{O}((M_1+M_2)^3)$ is described in Appendix~\ref{app:algorithm}.
From the real vectors $\vec{v}_i$ and $\vec{w}_j$ the observables can be obtained using representants
of a Clifford algebra, see~\cite{Tsirelson1993}.\\
In the following we provide two sufficient criteria for Ineq.~(\ref{eq:bound}) being tight.
\begin{corollary}\label{cor:quernormierung}
 If
\begin{align}
 ||V^d_{i,*}||=&\sqrt{\frac{d}{M_1}}\phantom{{},{}} \;\forall i\in\{1,2,...,M_1\} \\
\mbox{and } ||W^d_{j,*}||=&\sqrt{\frac{d}{M_2}}, \;\forall j\in\{1,2,...,M_2\}
\end{align}
then the bound is tight.
\end{corollary}
\begin{proof}
 The matrix $\alpha=\frac{M_1}{d}\1^d$ solves the system of equations (\ref{eq:normv}) and (\ref{eq:normw}).
\end{proof}
A second corollary treats the special case when $g$ is a square matrix and all singular values are the same.
\begin{corollary}\label{cor:equalsingularvalues}
 If $d=M_1=M_2$, then Ineq.~(\ref{eq:bound}) is tight.
\end{corollary}
\begin{proof}
Due to the orthogonality of $V$ and $W$, $\alpha=\1^d$ solves the system of equations (\ref{eq:normv}) and (\ref{eq:normw}).
\end{proof}
An application of this corollary is illustrated in the following example.
\begin{example}\label{ex:multipleCHSH}
Inequalities with coefficients
\begin{equation}
g= \left(
\begin{array}{cc}
 1 &  1 \\
 1 & -1
\end{array}
\right)^{\otimes k}
\end{equation}
are considered in~\cite{Heydari2006}, where for $k=2$ an upper bound of $4\sqrt{10}$ for the quantum value is given.
Ineq.~(\ref{eq:bound}) improves this bound to $T=8$, which coincides with the local realistic bound $B$. Note that
Corollary~\ref{cor:equalsingularvalues} states, that the bound $T(k)=2^{3 k/2}$ is tight for all $k$. It can be easily seen, that for all
even $k$, the classical value coincides with the quantum bound, i.e. the inequality is no Bell inequality. For odd $k$ numerical evidence
indicates, that the violation vanishes. Therefore we do not expect the violation to reach $Q/B=\sqrt{3}$ in the limit of large
$k$, different to the conjecture in~\cite{Heydari2006}. A value of $Q/B=\sqrt{3}$ would be near the maximal violation (Grothendieck's constant)
for any bipartite full correlation Bell experiment~\cite{Fishburn1994}. Please note that the well-known CHSH inequality is incorporated as
the special case with $k=1$.
\end{example}
Several more examples are given in Appendix~\ref{app:examples}, amongst them the famous CHSH
inequality~\cite{Clauser1969}
(example~\ref{ex:chsh}) and inequalities by Braunstein and Caves~\cite{Braunstein1990} (example~\ref{ex:BraunsteinCaves}), Vertesi and
P\'al~\cite{Vertesi2008} (example~\ref{ex:Vertesi}), Gisin~\cite{Gisin1999} (example~\ref{ex:GEQ}) and Fishburn and
Reeds~\cite{Fishburn1994} (example~\ref{ex:FishburnReeds}).\\
\indent The presented method can be generalized to more than two parties. All $n$-party Bell inequalities considered here are of the form
\begin{equation}
 \sum_{x_1,...,x_n=1}^{M_1,...,M_n} g(x_1,...,x_n) E_{\text{lr}}(x_1,...,x_n)\leq B.
\end{equation}
Each party $i$ receives a subsystem from the source and measures it in a setting $x_i\in\{1,2,...,M_i\}$. Suppose a time order
such that all but the first two parties do this before party one and two. Then the setup is exactly the same as
considered before, where the bipartite state is obtained by tracing out parties three to $n$. Formalizing
this one sees that
\begin{equation}
 \sum_{x_1,...,x_n=1}^{M_1,...,M_n} g(x_1,...,x_n) E(x_1,...,x_n) \leq T\\[1ex]
\end{equation}
with
\begin{equation}
 T=\sqrt{M_1 M_2} \sum_{x_3,...,x_n}^{M_3,...,M_n} ||g_{*,*,x_3,...,x_n}||_2. \label{eq:multipartite}
\end{equation}
Here $g_{*,*,x_3,...,x_n}$ denotes the matrix found in the $n$th-order tensor $g$ by fixing all but the first two indices. In general labeling
different parties as 1 and 2 leads to different values of the bound.
\begin{example}[Mermin-Inequality]\label{ex:mermin}
 The Mermin-Inequality is given by coefficients
\begin{equation}
  g(x_1,...,x_n)=\cos\left(\frac{\pi}{2}(x_1+x_2+...+x_n)\right).
\end{equation}
Eq.~(\ref{eq:multipartite}) gives the bound
\begin{equation}
 T=2 \sum_{x_3,...,x_n} \underbrace{||g_{*,*,x_3,...,x_n}||_2}_{=1} = 2^{n-1},
\end{equation}
which is achievable with a GHZ-state~\cite{Mermin1990}. Thus the bound is tight for this family of inequalities.
\end{example}
The insights on the mathematical structure gained above help to construct new Bell inequalities. We focus on the minimal
dimension of the involved observables required for the maximal violation. The dimension $d$ of the real vectors $\vec{v}_i$ and $\vec{w}_j$
is linked to the dimension of the corresponding observables $D$. Due to the explicit construction of observables in
Ref.~\cite{Tsirelson1993}, we know that 
\begin{equation}
 D\leq 2^{\lfloor d/2\rfloor}
\end{equation}
is possible, while it is also known~\cite{Vertesi2009}, that
\begin{equation}
 D\geq\left\lceil \frac{d+1}{2} \right\rceil
\end{equation}
is necessary.
We construct $g$ such that Eqs.~(\ref{eq:normv}) and (\ref{eq:normw}) are fulfilled for some matrix $\alpha$ with rank $d$. This implies
that the maximal violation can be achieved using $d$-dimensional real vectors $\vec{v}_i$ and $\vec{w}_j$. If in some experiment only
qubits ($D=2$) are available, then one can construct Bell inequalities with $d\leq 3$, assuring that the maximal violation is within the
scope of this experiment. This can be done by explicitly constructing the singular value decomposition of $g$, e.g.
\begin{equation}
 g= V \text{diag }\{2,2,2,1,...,1\} W^T, \label{eq:qubitconstruction}
\end{equation}
where $V$ and $W$ are unitary matrices, such that the conditions of Theorem~\ref{thm:tightness}, Corollary~\ref{cor:quernormierung} or
Corollary~\ref{cor:equalsingularvalues} are fulfilled.
\begin{example}[Inequality for qubits]\label{ex:qubits}
Consider the Bell inequality corresponding to a matrix $g$ given by Eq.~(\ref{eq:qubitconstruction}) for
\begin{align}
V=&\left(
\begin{array}{cc}
 \frac{1}{\sqrt{2}} & \frac{1}{\sqrt{2}} \\
 \frac{1}{\sqrt{2}} & -\frac{1}{\sqrt{2}}
\end{array}
\right)^{\otimes 2}\\
\mbox{and }W=&\left(
\begin{array}{cc}
 1 & 0 \\
 0 & -1
\end{array}
\right)\otimes\left(
\begin{array}{cc}
 \frac{1}{\sqrt{2}} & \frac{1}{\sqrt{2}} \\
 \frac{1}{\sqrt{2}} & -\frac{1}{\sqrt{2}}
\end{array}
\right).
\end{align}
By construction, the maximal quantum value is $Q=8$, while $B=4\sqrt{2}$ is the maximum achievable value 
within local hidden variable theories. From Eq.~(\ref{eq:qubitconstruction}) we know that $d=3$ and the maximal violation is achievable with qubits. Note that the singular value $S_{44}=1$ needs only to be smaller than $||g||_2=2$, i.e. it can also be chosen to be $0$.
\end{example}
Furthermore one might be interested in constructing Bell inequalities that cannot be violated by systems with dimension smaller than some
chosen dimension. Such Bell inequalities are a recent development called dimension
witnesses~\cite{Vertesi2008,Vertesi2009,PhysRevLett.100.210503,PhysRevLett.105.230501,PhysRevLett.102.190504,PhysRevA.78.062112,
PhysRevLett.111.030501}. Here the
unitary matrices $V$ and $W$ are constructed such that a rank $d$ solution $\alpha$ exists, but not
a rank $d-1$ solution. 
We can assume $\alpha$ to be a symmetric $d\times d$-matrix (see Appendix~\ref{app:algorithm}), i.e. $\alpha$ contains $ d (d+1)/2$ degrees
of
freedom. Therefore $d (d+1)/2$ vectors, that lie on a $d$-dimensional ellipsoid with finite semi-axes and lead to independent equations
(\ref{eq:normv}) or (\ref{eq:normw}), determine $\alpha$ and thus also it's rank to be $d$. 
 Note that the rows of both $V^d$ and $W^d$ form this set of vectors. The following simple construction illustrates this method.
\begin{example}[Random Dimension Witness]
 Given $d\in\mathds{N}$ greater or equal two, let $k=\lfloor (d-1)/2 \rfloor+1$ and $U_i$, $i\in\{1,..,k\}$, be random unitary $d\times d$
matrices. The inequality with coefficients given by the following $k d \times d$-matrix
\begin{equation}
 g=\left(\begin{array}{c}
          U_1\\
          U_2\\
\vdots\\
U_{k}
         \end{array}\right) \label{eq:randomdimwit}
\end{equation} 
corresponds to a Bell inequality. Note that the truncated singular value decomposition of $g$ can be read from Eq.~(\ref{eq:randomdimwit})
as $V^d=\frac{1}{\sqrt{k}} g$, $S^d=\sqrt{k} \1^d$, $W^d=\1^d$. The maximal quantum value $Q=k d$ is achievable
(Cor.~\ref{cor:quernormierung}). With probability one, the $k d$ measurement directions of party one and the $d$ measurement directions of
party two uniquely define a $d$-dimensional ellipsoid. Note that due to the orthogonality of $U_i$, more than $d (d+1)/2$ measurement
directions are used. Observables corresponding to real vectors spanning a space with dimension smaller than $d$
do not suffice to observe a maximal violation of such a Bell inequality and therefore it can be used as a dimension witness. The
number of measurement settings needed to witness dimension $d$ with this method is only $\mathcal{O}(d^2)$, while it is $\mathcal{O}(2^d)$ for the witness proposed in \cite{Vertesi2008},
see example~\ref{ex:Vertesi} in Appendix~\ref{app:examples}.
\end{example} 
In conclusion we introduced an approach for calculating upper bounds on the quantum value of correlation type Bell inequalities.
Computing the bound only requires the principal singular value of the coefficient matrix. We described how the tightness of
the bound can be tested. If the bound is reachable, which we find in several important examples, this method leads to optimal observables in
a natural way. Reversely, we showed how understanding the optimality conditions for our bound allows to construct Bell inequalities with
chosen properties, in particular properties of optimal observables, including their dimension.\\
The tools developed here may be useful to construct Bell inequalities with stronger violations than the known inequalities for
this scenario. Amongst other advantages, this may help to close the detection loophole in Bell test experiments. Furthermore, an improved
generalization of the bound for three and more parties is possibly of avail.
\begin{acknowledgments}
We thank Costantino Budroni, Otfried Gühne and Tobias Moroder for helpful discussions. M.E. is supported by Deutsche
Forschungsgemeinschaft
(DFG).
\end{acknowledgments}

\bibliographystyle{apsrev4-1} 
\input{arxiv.bbl}

\appendix
\section{Proof of Theorem~\ref{thm:tightness}}\label{app:proofoftightness}
From the proof of Theorem~\ref{thm:bound} we know that
\begin{equation}
\sum_{x_1=1}^{M_1} \sum_{x_2=1}^{M_2} g_{x_1,x_2} E(x_1,x_2) = \vec{V}^T (g \otimes \1^{M_1+M_2}) \vec{W}
\end{equation}
where the real vectors $\vec{V}$ and $\vec{W}$ are defined in Eq.~(\ref{eq:defVundW}). From this we see, that the bound is reached, if and
only if
$\vec{V}$ and $\vec{W}$ are ``matching'' singular vectors to the maximal singular value, i.e. $(g\otimes\1^{M_1+M_2})
\vec{W}=\sqrt{M_2/M_1} ||g||_2 \vec{V}$, while at the same time the respective vectors $\vec{v}_i$ and $\vec{w}_j$ are unit vectors. The
normalization of $\vec{v}_i$ and
$\vec{w}_j$ is required by Tsirelson's theorem. General singular vectors to the maximal singular value can be written as
\begin{align}
\vec{V}=&\sum_{l_1=1}^{d} \sum_{l_2=1}^{M_1+M_2} \alpha_{l_1,l_2}  V_{*,l_1} \otimes
\1^{M_1+M_2}_{*,l_2},\label{eq:Vvortransform}\\
\vec{W}=&\sum_{l_1=1}^{d} \sum_{l_2=1}^{M_1+M_2} \beta_{l_1,l_2}  W_{*,l_1} \otimes
\1^{M_1+M_2}_{*,l_2},\label{eq:Wvortransform}
\end{align}
where $V_{*,l_1}$ denotes the $l_1$-th row of the matrix $V$ as a column vector and $\1^{M_1+M_2}_{*,l_2}$ denotes the $l_2$-th canonical
basis vector. The fact that $\vec{V}$ matches $\vec{W}$ becomes manifest in
\begin{equation}
 \alpha_{l_1,l_2}=\sqrt{\frac{M_1}{M_2}} \beta_{l_1,l_2}. \label{eq:alphabeta}
\end{equation}
We are interested in the components $\alpha_{l_1,l_2}$ introduced in Eqs.~(\ref{eq:Vvortransform}) and (\ref{eq:Wvortransform}). They are
restricted by the norm conditions for $\vec{v}_i$ and $\vec{w}_j$, which read
\begin{align}
 1=&||\vec{v}_i||^2=||{\alpha}^T V^d_{i,*} ||^2 \label{eq:normvialpha}\\
\mbox{and } 1=&||\vec{w}_i||^2=\sqrt{\frac{M_2}{M_1}}||{\alpha}^T W^d_{i,*} ||^2.\label{eq:normwibeta}
\end{align}
Therefore the bound is tight, if and only if this system of equations is solvable. We conclude by showing, how the number of columns of
$\alpha$ is related to the dimension of the measurement vectors. If and only if the bound is reachable with $d'$-dimensional vectors
$\vec{v}_i$, $\vec{w}_j$, the system of equations is solvable by a $d\times d'$-matrix $\alpha$, where $d'\leq d$.
\begin{description}
 \item[``$\Leftarrow$''] 
If $\alpha$ is a $d\times d'$-matrix that solves the system of equations, then 
\begin{align}
 d'\geq \text{rank}\,\alpha \geq \text{dim}\, \text{span}\{v_i,w_j\} ,
\end{align}
where the last $\geq$-sign holds because $\vec{v}_i=\alpha^T V^d_{i*}$ and $\vec{w}_i=\sqrt{\frac{M_2}{M_1}}\alpha^T W^d_{i*}$, i.e.
$\vec{v}_i$ and $\vec{w}_j$ lie in the image of $\alpha^T$. The result $\text{dim}\, \text{span}\{v_i,w_j\}\leq d'$ implies, that after some
appropriate rotation, $(\vec{v}_i)_k=0$ and $(\vec{w}_j)_l=0$ for $k,l>d'$ and all $i,j$. Therefore $\vec{v}_i$ and $\vec{w}_j$ can be
considered to be elements of $\mathds{R}^{d'}$. Observables associated with these $d'$-dimensional vectors permit maximal violation.
\item[``$\Rightarrow$''] If the bound is reachable with $d'$-dimensional vectors $\vec{v}_i$, $\vec{w}_j$, then all vectors $\vec{v}_i$ and
$\vec{w}_j$ lie on a
$d'$-dimensional unit sphere. Without affecting the mapping of the $\vec{v}_i$ and $\vec{w}_j$, the image of $\alpha$ can be chosen to
coincide with the $d'$ dimensional subspace spanned by $\vec{v}_i$ and $\vec{w}_j$, so the rank of $\alpha$ can be chosen to be $d'$. The
rank is equal to the number of nonzero singular values. The truncated singular value decomposition associated with all nonzero singular
values equals $\alpha$. Let us call it $\alpha=\tilde{V} \tilde{S} \tilde{W}^T$, so $\alpha
\alpha^T=\tilde{V} \tilde{S} \tilde{W}^T \tilde{W} \tilde{S} \tilde{V}^T=\tilde{V} \tilde{S} \tilde{S} \tilde{V}^T$, therefore
$\alpha'=\tilde{V} \tilde{S}$ is a $d\times d'$-matrix solving the system of equations. 
\end{description}
\section{Algorithm to find $\alpha$}\label{app:algorithm}
We want to find the solution $\alpha$ to the system of equations
\begin{align}
 ||\alpha^T V^d_{i,*} ||^2=&1 &\forall i\in\{1,2,...,M_1\} \label{eq:normv2}\\
 ||\alpha^T W^d_{j,*} ||^2=&\frac{M_1}{M_2}&\forall j\in\{1,2,...,M_2\}\label{eq:normw2}.
\end{align}
It is convenient to rewrite these equations as
\begin{align}
 A_{i,*}^T X A_{i,*}=&1 &\forall i\in\{1,2,...,M_1+M_2\} \label{eq:normA}
\end{align}
where $X=\alpha \alpha^T$ is unknown and 
\begin{equation}
A=\left(\begin{array}{c} V^d\\ \sqrt{\frac{M_2}{M_1}} W^d\end{array}\right)
\end{equation}
is a $(M_1+M_2)\times d$ matrix containing all the vectors $A_{i,*}$ which will be normalized after application of $\alpha^T$ (if
possible). Eq.~(\ref{eq:normA}) restricts $X$ on the space spanned by these vectors. Unaffected by their linear dependence, the unknown $X$
in Eq.~(\ref{eq:normA}) can be defined via it's action on these vectors,
\begin{equation}
 X A_{i,*}=\sum_k \tilde{c}_{i,k} A_{k,*}, \label{eq:actionofX}
\end{equation}
where $\tilde{c}_{i,k}$ is real. If $A_{i,*}$ and $A_{k,*}$ are
perpendicular, then we can choose $\tilde{c}_{i,k}=0$. Thus we use the form
\begin{equation}
 \tilde{c}_{i,k} = A_{k,*}^T A_{i,*} c_{i,k} \label{eq:formctilde}
\end{equation}
and Eq.~(\ref{eq:actionofX}) becomes
\begin{align}
 X A_{i,*}=&\sum_k c_{i,k} A_{k,*}^T A_{i,*} A_{k,*}\\
          =&\underbrace{\sum_k c_{i,k} A_{k,*} A_{k,*}^T}_{X} A_{i,*},
\end{align}
from which we can read $X$. This has to be the same $X$ for every equation in the system of equations (\ref{eq:normA}), i.e. $c_{i,k}=c_k$. 
We have
\begin{align}
X=&\sum_{k=1}^{M_1+M_2} c_{k} A_{k,*} A_{k,*}^T\\
 =&A^T \text{\;diag\,}(c_1,...,c_{M_1+M_2} ) A.\label{eq:Xvonc}
\end{align}
 Inserting this into Eq.~(\ref{eq:normA}) gives for all $i$
\begin{align}
1=&(A X A^T)_{ii}\\
 =&(P \text{\;diag\,}(c_1,...,c_{M_1+M_2} ) P)_{ii}\\
 =&\sum_{k=1}^{M_1+M_2} c_k P_{ik}^2. \label{eq:sumk}
\end{align}
Here we introduced the projector $P=A A^T$. We also introduce the matrix $Q$, which is $P$ componentwise squared, i.e. $Q_{ij}=P_{ij}^2$,
 and the vector $\vec{1}$, where every component is one. Then Eq.~(\ref{eq:sumk}) can be written as
\begin{equation}
Q\vec{c}=\vec{1}
\end{equation}
This equation is solvable if and only if
\begin{equation}
\vec{1}=Q Q^-\vec{1},
\end{equation}
where $Q^-$ is the pseudoinverse of $Q$. Then all solutions to this equation are given by
\begin{equation}
\vec{c}=\underbrace{Q^-\vec{1}}_{\vec{c}_0}+\underbrace{(\1-Q^- Q)\vec{y}}_{\vec{c}_y},
\end{equation}
with $\vec{y}\in\R^{M_1+M_2}$. Here we marked the $y$-independent and $y$-dependent part of $\vec{c}$. Inserting into
 Eq.~(\ref{eq:Xvonc}) gives a $y$-independent part and a $y$-dependent part of $X$, i.e.
\begin{equation}
 X=X_0+X_y.
\end{equation}
The vector $\vec{c}_y=(\1-Q^- Q)\vec{y}$ lies in the kernel of $Q$. Therefore for all $i\leq M_1+M_2$
\begin{align}
 0&=\sum_{k=1}^{M_1+M_2} Q_{ik} (c_y)_k \\
&= \sum_{l_1,l_2=1}^d A_{i l_1} A_{i l_2} \underbrace{\sum_{k=1}^{M_1+M_2}  A_{k l_1} A_{k l_2} (c_y)_k}_{(X_y)_{l_1,l_2}}\\
&=A_{i,*}^T X_y A_{i,*}.
\end{align}
This implies that $X_y=0$ and thus $X=X_0$ is uniquely defined by Eq.~(\ref{eq:normA}) and Eq.~(\ref{eq:formctilde}). We obtain a solution
$\alpha$ with $\alpha \alpha^T = X$ via
\begin{equation}
\alpha=\sqrt{X}.
\label{eq:solutionalpha}
\end{equation}
It is possible, that $X$ is not semipositive, in which case there is no real solution $\alpha$.\\
The described algorithm contains a singular value decomposition, the calculation of a pseudoinverse and a square root of a matrix, as well
as several matrix multiplications. The runtime complexities of all of these operations are asymptotically upper bounded by the matrix
dimension to the power of three~\cite{MatrixComputations}. Therefore the runtime complexity of this algorithm is $\mathcal{O}( (M_1+M_2)^3
)$.\\
A summarized pseudo code version of the described algorithm follows.
\begin{algorithmic}[1]
\Procedure{AlphaMatrix}{$g$}
\State $(V,S,W)\gets \Call{SVD}{g}$ \Comment{singular value decomposition of $g$}
\State $d\gets \max_{i:S_{ii}= S_{11}}  i$ \Comment{degeneracy of maximal singular value}
\State $V^d\gets V$ with columns $d+1$ to $M_1$ dropped \Comment{truncated SVD}
\State $W^d\gets W$ with columns $d+1$ to $M_2$ dropped
\State $A\gets \left(\begin{array}{c} V^d \\ \sqrt{\frac{M_2}{M_1}} W^d \end{array} \right)$ \Comment{the set of vectors on ellipsoid}
\State $P\gets A A^T$
\ForAll{$i,j\in\{1,2,...,M_1+M_2\}$} 
 \State $Q_{i,j} \gets P_{i,j}^2$
\EndFor
\State $\vec{c}\gets Q^- \vec{1}$ \Comment{apply pseudoinverse}
\If{$Q \vec{c} = \vec{1}$} \Comment{solution exists}
  \State $X\gets A^T \text{\;diag\,}(\vec{c}) A$
  \State $\alpha\gets \sqrt{X}$
  \If{Im($\alpha$)=0}
   \State\Return $\alpha$
  \Else \Comment{$X$ is not semipositive}
   \State\Return $0$ \Comment{only complex solutions}
  \EndIf
\Else
 \State\Return $0$ \Comment{equation not solvable}
\EndIf
\EndProcedure
\end{algorithmic}

\section{A collection of instructive examples}\label{app:examples}
This section contains more examples.
\begin{example}[CHSH inequality]\label{ex:chsh}
 The original Clauser-Horne-Shimony-Holt(CHSH) inequality~\cite{Clauser1969} is given by
\begin{equation}
 g=\left(\begin{array}{cc}
        1 &  1\\
        1 & -1
       \end{array}\right).
\end{equation}
As $g$ is symmetric, the singular values are given by the absolute values of it's eigenvalues, which is $\sqrt{2}$. Since all singular
values are equal, the bound in Ineq.~(\ref{eq:bound}) is tight (Corollary~\ref{cor:equalsingularvalues}, see also Fig.~\ref{fig:CHSH}).
It is the well-known upper bound of $T=Q=2\sqrt{2}$ for the quantum value of the CHSH-inequality derived by Tsirelson~\cite{Cirel'son1980}.
\end{example}
\begin{example}\label{ex:nottight}
 Consider the coefficients
\begin{align}
 g=& \left(
\begin{array}{cc}
 1 & 1 \\
 1 & 0 \\
 \end{array}
\right),
\end{align}
where the bound gives $T=1+\sqrt{5}$, but obviously only $3$ can be reached in any theory. Therefore the bound is not tight for this
instance of $g$.
\end{example}
\begin{example}[Binary digits]\label{ex:Vertesi}
 In Ref.~\cite{Vertesi2008} a bipartite Bell inequality given by coefficients
\begin{equation}
 g_{x_1,x_2}= 1-2(\lfloor 2^{1-x_2} (x_1-1) \rfloor \bmod 2)
\end{equation}
is discussed, which resembles a list of binary numbers. The number of measurement settings is given by $M_1=2^{M_2-1}$. It can be used to
witness observables referring to $d=M_2$ dimensional real vectors. Thus the number of measurement settings $M_1+M_2$ is $\mathcal{O}(2^d)$.
Bounds on the value of the Bell inequality are given in the reference.\\
It can be shown, that all singular values of $g$ are equal to $\sqrt{M_1}=\sqrt{2^{M_2-1}}$. A singular value decomposition of $g$ then is
\begin{equation}
 g =  \underbrace{\frac{1}{\sqrt{M_1}} g}_V \underbrace{\sqrt{M_1}\1^{M_2}}_S \underbrace{\1^{M_2}}_{W^T}.
\end{equation}
From this the diagonal solution $\alpha_i=\sqrt{M_1}{M_2}$ can be read. This implies, that the bound $T=M_1\sqrt{M_2}$ is tight.
\end{example}

\begin{example}[Braunstein-Caves inequalities]\label{ex:BraunsteinCaves}
The Braunstein-Caves inequalities~\cite{Braunstein1990} are given by
\begin{equation}
g_{x_1,x_2}=\left\{ \begin{array}{cl}
1 & \text{if } 0\leq x_1-x_2\leq 1 \\
-1 & \text{if } x_1=1 \text{ and } x_2=M \\
0 & \text{else}
\end{array}\right.,
\end{equation}
where $M=M_1=M_2$.
It can be shown that the maximal singular value of $g$ is $2 \cos(\pi/(2 M))$ and twofold degenerate. The bound reads $T=2
M\cos(\pi/(2 M))$, which is achievable~\cite{Braunstein1990,Cabello2009}. See also Fig.~\ref{fig:BC}.
\end{example}

\begin{example}[Greater Equal Function]\label{ex:GEQ}
The greater-equal-function is related to a Bell inequality with coefficients
\begin{equation}
 g_{x_1,x_2}=\left\{\begin{array}{cl}
                     1 & \text{if } x_1\geq x_2\\
                    -1 & \text{else}
                    \end{array}\right.,
\end{equation}
where $1\leq x_1,x_2\leq M_1=M_2=M$~\cite{Gisin1999}. The maximal singular value of $g$, $\csc(\pi/(2 M))$, is twofold
degenerate. The quantum bound $T=M\csc(\pi/(2 M)$ is tight (Corollary~\ref{cor:quernormierung}, see also Fig.~\ref{fig:GEQ}) and
strictly larger than the local hidden variable bound $B=\lceil M^2/2\rceil$. The violation $Q/B$ in the
limit of large $M$ is $4/\pi$~\cite{Gisin1999}.
\end{example}
\begin{figure}[tbp]
 \subfigure[CHSH Inequality (Ex.~\ref{ex:chsh})]{\includegraphics[width=0.3\linewidth]{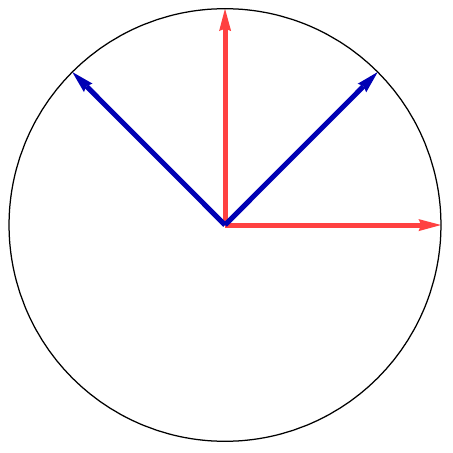} \label{fig:CHSH}}
 \subfigure[BC Inequality (Ex.~\ref{ex:BraunsteinCaves})]{\includegraphics[width=0.3\linewidth]{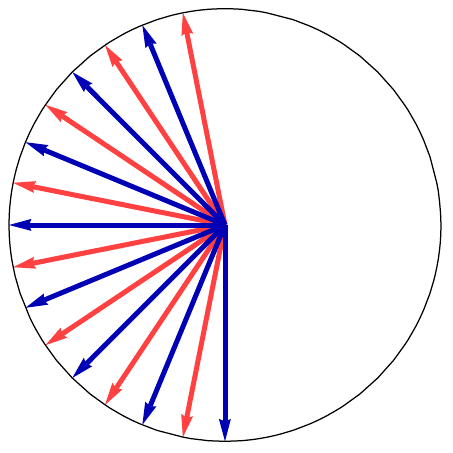} \label{fig:BC}}
 \subfigure[Gisin's Inequality (Ex.~\ref{ex:GEQ})]{\includegraphics[width=0.3\linewidth]{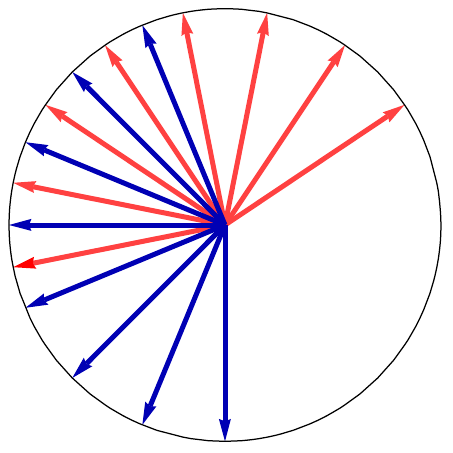} \label{fig:GEQ}}
 \caption{If and only if the bound is tight, the vectors $V^d_{i,*}$ (blue) and $\sqrt{\frac{M_2}{M_1}} W^d_{j,*}$ (red) lie on the
surface of an origin-centered Ellipsoid.}
 \label{fig:ellipsen}
\end{figure}
\begin{example}[Fishburn-Reeds]\label{ex:FishburnReeds}
The highest violation of an explicit bipartite correlation type Bell inequality known to the authors is given by Fishburn and Reeds
in~\cite{Fishburn1994}. They describe a series of Bell inequalities, which is constructed as follows. Construct a $k(k-1)\times k$-matrix
$F_k$, which rows constitute all vectors of the form $(0,...,0,-1,0,...,0,1,0,...,0)$ and $(0,...,0,1,0,...,0,1,0,...,0)$. The Bell
inequality is given by
coefficients
\begin{equation}
 g=F_k F_k^T-\frac{4}{3}\1.
\end{equation}
By construction, $g'=F_k F_k^T$ fulfills the conditions of Corollary~\ref{cor:quernormierung}. The diagonal modification changes the
singular values, without changing their order. Therefore also $g$ fulfills the conditions of Corollary~\ref{cor:quernormierung}. Because the
maximal singular value is $2(k-1)-4/3$, the maximal quantum value is $Q=T=(2 (k - 1) - 4/3) k (k - 1)$, which is the value derived in the
reference. The first $k$ for which $Q/B>\sqrt{2}$ is $k=5$, where $Q/B=\frac{10}{7}\approx 1.42857$. For $k=5$, the explicit form of $g$ is
\begin{widetext}
\begin{equation}
 g=\left( 
\begin{array}{cccccccccccccccccccc}
 \frac{2}{3} & 1 & 1 & 1 & 1 & 1 & 1 & 0 & 0 & 0 & 0 & 1 & 1 & 1 & 1 & 1 & 1 & 0 & 0 & 0 \\
 1 & \frac{2}{3} & 1 & 1 & 1 & 0 & 0 & 1 & 1 & 0 & 1 & 0 & 1 & 1 & -1 & 0 & 0 & 1 & 1 & 0 \\
 1 & 1 & \frac{2}{3} & 1 & 0 & 1 & 0 & 1 & 0 & 1 & 1 & 1 & 0 & 1 & 0 & -1 & 0 & -1 & 0 & 1 \\
 1 & 1 & 1 & \frac{2}{3} & 0 & 0 & 1 & 0 & 1 & 1 & 1 & 1 & 1 & 0 & 0 & 0 & -1 & 0 & -1 & -1 \\
 1 & 1 & 0 & 0 & \frac{2}{3} & 1 & 1 & 1 & 1 & 0 & -1 & -1 & 0 & 0 & 0 & 1 & 1 & 1 & 1 & 0 \\
 1 & 0 & 1 & 0 & 1 & \frac{2}{3} & 1 & 1 & 0 & 1 & -1 & 0 & -1 & 0 & 1 & 0 & 1 & -1 & 0 & 1 \\
 1 & 0 & 0 & 1 & 1 & 1 & \frac{2}{3} & 0 & 1 & 1 & -1 & 0 & 0 & -1 & 1 & 1 & 0 & 0 & -1 & -1 \\
 0 & 1 & 1 & 0 & 1 & 1 & 0 & \frac{2}{3} & 1 & 1 & 0 & -1 & -1 & 0 & -1 & -1 & 0 & 0 & 1 & 1 \\
 0 & 1 & 0 & 1 & 1 & 0 & 1 & 1 & \frac{2}{3} & 1 & 0 & -1 & 0 & -1 & -1 & 0 & -1 & 1 & 0 & -1 \\
 0 & 0 & 1 & 1 & 0 & 1 & 1 & 1 & 1 & \frac{2}{3} & 0 & 0 & -1 & -1 & 0 & -1 & -1 & -1 & -1 & 0 \\
 0 & 1 & 1 & 1 & -1 & -1 & -1 & 0 & 0 & 0 & \frac{2}{3} & 1 & 1 & 1 & -1 & -1 & -1 & 0 & 0 & 0 \\
 1 & 0 & 1 & 1 & -1 & 0 & 0 & -1 & -1 & 0 & 1 & \frac{2}{3} & 1 & 1 & 1 & 0 & 0 & -1 & -1 & 0 \\
 1 & 1 & 0 & 1 & 0 & -1 & 0 & -1 & 0 & -1 & 1 & 1 & \frac{2}{3} & 1 & 0 & 1 & 0 & 1 & 0 & -1 \\
 1 & 1 & 1 & 0 & 0 & 0 & -1 & 0 & -1 & -1 & 1 & 1 & 1 & \frac{2}{3} & 0 & 0 & 1 & 0 & 1 & 1 \\
 1 & -1 & 0 & 0 & 0 & 1 & 1 & -1 & -1 & 0 & -1 & 1 & 0 & 0 & \frac{2}{3} & 1 & 1 & -1 & -1 & 0 \\
 1 & 0 & -1 & 0 & 1 & 0 & 1 & -1 & 0 & -1 & -1 & 0 & 1 & 0 & 1 & \frac{2}{3} & 1 & 1 & 0 & -1 \\
 1 & 0 & 0 & -1 & 1 & 1 & 0 & 0 & -1 & -1 & -1 & 0 & 0 & 1 & 1 & 1 & \frac{2}{3} & 0 & 1 & 1 \\
 0 & 1 & -1 & 0 & 1 & -1 & 0 & 0 & 1 & -1 & 0 & -1 & 1 & 0 & -1 & 1 & 0 & \frac{2}{3} & 1 & -1 \\
 0 & 1 & 0 & -1 & 1 & 0 & -1 & 1 & 0 & -1 & 0 & -1 & 0 & 1 & -1 & 0 & 1 & 1 & \frac{2}{3} & 1 \\
 0 & 0 & 1 & -1 & 0 & 1 & -1 & 1 & -1 & 0 & 0 & 0 & -1 & 1 & 0 & -1 & 1 & -1 & 1 & \frac{2}{3}
\end{array}
\right).
\end{equation}
\end{widetext}
\end{example}
\end{document}

%% file: arxiv.bbl
%

%% file: arxiv.bbl
\begin{thebibliography}{41}%
\makeatletter
\providecommand \@ifxundefined [1]{%
 \@ifx{#1\undefined}
}%
\providecommand \@ifnum [1]{%
 \ifnum #1\expandafter \@firstoftwo
 \else \expandafter \@secondoftwo
 \fi
}%
\providecommand \@ifx [1]{%
 \ifx #1\expandafter \@firstoftwo
 \else \expandafter \@secondoftwo
 \fi
}%
\providecommand \natexlab [1]{#1}%
\providecommand \enquote  [1]{``#1''}%
\providecommand \bibnamefont  [1]{#1}%
\providecommand \bibfnamefont [1]{#1}%
\providecommand \citenamefont [1]{#1}%
\providecommand \href@noop [0]{\@secondoftwo}%
\providecommand \href [0]{\begingroup \@sanitize@url \@href}%
\providecommand \@href[1]{\@@startlink{#1}\@@href}%
\providecommand \@@href[1]{\endgroup#1\@@endlink}%
\providecommand \@sanitize@url [0]{\catcode `\\12\catcode `\$12\catcode
  `\&12\catcode `\#12\catcode `\^12\catcode `\_12\catcode `\%12\relax}%
\providecommand \@@startlink[1]{}%
\providecommand \@@endlink[0]{}%
\providecommand \url  [0]{\begingroup\@sanitize@url \@url }%
\providecommand \@url [1]{\endgroup\@href {#1}{\urlprefix }}%
\providecommand \urlprefix  [0]{URL }%
\providecommand \Eprint [0]{\href }%
\providecommand \doibase [0]{http://dx.doi.org/}%
\providecommand \selectlanguage [0]{\@gobble}%
\providecommand \bibinfo  [0]{\@secondoftwo}%
\providecommand \bibfield  [0]{\@secondoftwo}%
\providecommand \translation [1]{[#1]}%
\providecommand \BibitemOpen [0]{}%
\providecommand \bibitemStop [0]{}%
\providecommand \bibitemNoStop [0]{.\EOS\space}%
\providecommand \EOS [0]{\spacefactor3000\relax}%
\providecommand \BibitemShut  [1]{\csname bibitem#1\endcsname}%
\let\auto@bib@innerbib\@empty
\bibitem [{\citenamefont {Bell}(1964)}]{Bell1964}%
  \BibitemOpen
  \bibfield  {author} {\bibinfo {author} {\bibfnamefont {J.~S.}\ \bibnamefont
  {Bell}},\ }\href@noop {} {\bibfield  {journal} {\bibinfo  {journal} {Physics
  NY}\ }\textbf {\bibinfo {volume} {1}},\ \bibinfo {pages} {195} (\bibinfo
  {year} {1964})}\BibitemShut {NoStop}%
\bibitem [{\citenamefont {Weihs}\ \emph {et~al.}(1998)\citenamefont {Weihs},
  \citenamefont {Jennewein},\ and\ \citenamefont {Simon}}]{Weihs1998}%
  \BibitemOpen
  \bibfield  {author} {\bibinfo {author} {\bibfnamefont {G.}~\bibnamefont
  {Weihs}}, \bibinfo {author} {\bibfnamefont {T.}~\bibnamefont {Jennewein}}, \
  and\ \bibinfo {author} {\bibfnamefont {C.}~\bibnamefont {Simon}},\
  }\href@noop {} {\bibfield  {journal} {\bibinfo  {journal} {Phys. Rev. Lett.}\
  }\textbf {\bibinfo {volume} {81}},\ \bibinfo {pages} {5039} (\bibinfo {year}
  {1998})}\BibitemShut {NoStop}%
\bibitem [{\citenamefont {Clauser}\ \emph {et~al.}(1969)\citenamefont
  {Clauser}, \citenamefont {Horne}, \citenamefont {Shimony},\ and\
  \citenamefont {Holt}}]{Clauser1969}%
  \BibitemOpen
  \bibfield  {author} {\bibinfo {author} {\bibfnamefont {J.}~\bibnamefont
  {Clauser}}, \bibinfo {author} {\bibfnamefont {M.}~\bibnamefont {Horne}},
  \bibinfo {author} {\bibfnamefont {A.}~\bibnamefont {Shimony}}, \ and\
  \bibinfo {author} {\bibfnamefont {R.}~\bibnamefont {Holt}},\ }\href@noop {}
  {\bibfield  {journal} {\bibinfo  {journal} {Phys. Rev. Lett.}\ }\textbf
  {\bibinfo {volume} {23}},\ \bibinfo {pages} {880} (\bibinfo {year}
  {1969})}\BibitemShut {NoStop}%
\bibitem [{\citenamefont {Mermin}(1990)}]{Mermin1990}%
  \BibitemOpen
  \bibfield  {author} {\bibinfo {author} {\bibfnamefont {N.}~\bibnamefont
  {Mermin}},\ }\href@noop {} {\bibfield  {journal} {\bibinfo  {journal} {Phys.
  Rev. Lett.}\ }\textbf {\bibinfo {volume} {65}},\ \bibinfo {pages} {1838}
  (\bibinfo {year} {1990})}\BibitemShut {NoStop}%
\bibitem [{\citenamefont {Braunstein}\ and\ \citenamefont
  {Caves}(1990)}]{Braunstein1990}%
  \BibitemOpen
  \bibfield  {author} {\bibinfo {author} {\bibfnamefont {S.}~\bibnamefont
  {Braunstein}}\ and\ \bibinfo {author} {\bibfnamefont {C.}~\bibnamefont
  {Caves}},\ }\href@noop {} {\bibfield  {journal} {\bibinfo  {journal} {Ann.
  Phys.-NY}\ }\textbf {\bibinfo {volume} {202}},\ \bibinfo {pages} {22}
  (\bibinfo {year} {1990})}\BibitemShut {NoStop}%
\bibitem [{\citenamefont {Cavalcanti}\ \emph {et~al.}(2007)\citenamefont
  {Cavalcanti}, \citenamefont {Foster}, \citenamefont {Reid},\ and\
  \citenamefont {Drummond}}]{Cavalcanti2007}%
  \BibitemOpen
  \bibfield  {author} {\bibinfo {author} {\bibfnamefont {E.~G.}\ \bibnamefont
  {Cavalcanti}}, \bibinfo {author} {\bibfnamefont {C.~J.}\ \bibnamefont
  {Foster}}, \bibinfo {author} {\bibfnamefont {M.~D.}\ \bibnamefont {Reid}}, \
  and\ \bibinfo {author} {\bibfnamefont {P.~D.}\ \bibnamefont {Drummond}},\
  }\href@noop {} {\bibfield  {journal} {\bibinfo  {journal} {Phys. Rev. Lett.}\
  }\textbf {\bibinfo {volume} {99}},\ \bibinfo {pages} {210405} (\bibinfo
  {year} {2007})}\BibitemShut {NoStop}%
\bibitem [{\citenamefont {V\'{e}rtesi}\ and\ \citenamefont
  {P\'{a}l}(2008)}]{Vertesi2008}%
  \BibitemOpen
  \bibfield  {author} {\bibinfo {author} {\bibfnamefont {T.}~\bibnamefont
  {V\'{e}rtesi}}\ and\ \bibinfo {author} {\bibfnamefont {K.}~\bibnamefont
  {P\'{a}l}},\ }\href {http://pra.apsorg/abstract/PRA/v77/i4/e042106}
  {\bibfield  {journal} {\bibinfo  {journal} {Phys. Rev. A}\ }\textbf {\bibinfo
  {volume} {77}},\ \bibinfo {pages} {042106} (\bibinfo {year}
  {2008})}\BibitemShut {NoStop}%
\bibitem [{\citenamefont {Pitowsky}\ and\ \citenamefont
  {Svozil}(2001)}]{PhysRevA.64.014102}%
  \BibitemOpen
  \bibfield  {author} {\bibinfo {author} {\bibfnamefont {I.}~\bibnamefont
  {Pitowsky}}\ and\ \bibinfo {author} {\bibfnamefont {K.}~\bibnamefont
  {Svozil}},\ }\href@noop {} {\bibfield  {journal} {\bibinfo  {journal} {Phys.
  Rev. A}\ }\textbf {\bibinfo {volume} {64}},\ \bibinfo {pages} {014102}
  (\bibinfo {year} {2001})}\BibitemShut {NoStop}%
\bibitem [{\citenamefont {Gisin}(1999)}]{Gisin1999}%
  \BibitemOpen
  \bibfield  {author} {\bibinfo {author} {\bibfnamefont {N.}~\bibnamefont
  {Gisin}},\ }\href
  {http://www.sciencedirect.com/science/article/pii/S0375960199004284}
  {\bibfield  {journal} {\bibinfo  {journal} {Phys. Lett. A}\ }\textbf
  {\bibinfo {volume} {260}},\ \bibinfo {pages} {8} (\bibinfo {year}
  {1999})}\BibitemShut {NoStop}%
\bibitem [{\citenamefont {Laskowski}\ \emph {et~al.}(2004)\citenamefont
  {Laskowski}, \citenamefont {Paterek}, \citenamefont {Zukowski},\ and\
  \citenamefont {Brukner}}]{CaslavMultiparty}%
  \BibitemOpen
  \bibfield  {author} {\bibinfo {author} {\bibfnamefont {W.}~\bibnamefont
  {Laskowski}}, \bibinfo {author} {\bibfnamefont {T.}~\bibnamefont {Paterek}},
  \bibinfo {author} {\bibfnamefont {M.}~\bibnamefont {Zukowski}}, \ and\
  \bibinfo {author} {\bibfnamefont {C.}~\bibnamefont {Brukner}},\ }\href@noop
  {} {\bibfield  {journal} {\bibinfo  {journal} {Phys. Rev. Lett.}\ }\textbf
  {\bibinfo {volume} {93}},\ \bibinfo {pages} {200401} (\bibinfo {year}
  {2004})}\BibitemShut {NoStop}%
\bibitem [{\citenamefont {Brukner}\ \emph {et~al.}(2004)\citenamefont
  {Brukner}, \citenamefont {Zukowski}, \citenamefont {Pan},\ and\ \citenamefont
  {Zeilinger}}]{Brukner2004}%
  \BibitemOpen
  \bibfield  {author} {\bibinfo {author} {\bibfnamefont {C.}~\bibnamefont
  {Brukner}}, \bibinfo {author} {\bibfnamefont {M.}~\bibnamefont {Zukowski}},
  \bibinfo {author} {\bibfnamefont {J.-W.}\ \bibnamefont {Pan}}, \ and\
  \bibinfo {author} {\bibfnamefont {A.}~\bibnamefont {Zeilinger}},\ }\href@noop
  {} {\bibfield  {journal} {\bibinfo  {journal} {Phys. Rev. Lett.}\ }\textbf
  {\bibinfo {volume} {92}},\ \bibinfo {pages} {127901} (\bibinfo {year}
  {2004})}\BibitemShut {NoStop}%
\bibitem [{\citenamefont {Ekert}(1991)}]{PhysRevLett.67.661}%
  \BibitemOpen
  \bibfield  {author} {\bibinfo {author} {\bibfnamefont {A.~K.}\ \bibnamefont
  {Ekert}},\ }\href@noop {} {\bibfield  {journal} {\bibinfo  {journal} {Phys.
  Rev. Lett.}\ }\textbf {\bibinfo {volume} {67}},\ \bibinfo {pages} {661}
  (\bibinfo {year} {1991})}\BibitemShut {NoStop}%
\bibitem [{\citenamefont {Barrett}\ \emph {et~al.}(2005)\citenamefont
  {Barrett}, \citenamefont {Hardy},\ and\ \citenamefont {Kent}}]{Barrett2005}%
  \BibitemOpen
  \bibfield  {author} {\bibinfo {author} {\bibfnamefont {J.}~\bibnamefont
  {Barrett}}, \bibinfo {author} {\bibfnamefont {L.}~\bibnamefont {Hardy}}, \
  and\ \bibinfo {author} {\bibfnamefont {A.}~\bibnamefont {Kent}},\ }\href@noop
  {} {\bibfield  {journal} {\bibinfo  {journal} {Phys. Rev. Lett.}\ }\textbf
  {\bibinfo {volume} {95}},\ \bibinfo {pages} {010503} (\bibinfo {year}
  {2005})}\BibitemShut {NoStop}%
\bibitem [{\citenamefont {Pironio}\ \emph {et~al.}(2010)\citenamefont
  {Pironio}, \citenamefont {Ac\'{\i}n}, \citenamefont {Massar}, \citenamefont
  {de~la Giroday}, \citenamefont {Matsukevich}, \citenamefont {Maunz},
  \citenamefont {Olmschenk}, \citenamefont {Hayes}, \citenamefont {Luo},
  \citenamefont {Manning},\ and\ \citenamefont {Monroe}}]{Pironio2010}%
  \BibitemOpen
  \bibfield  {author} {\bibinfo {author} {\bibfnamefont {S.}~\bibnamefont
  {Pironio}}, \bibinfo {author} {\bibfnamefont {A.}~\bibnamefont {Ac\'{\i}n}},
  \bibinfo {author} {\bibfnamefont {S.}~\bibnamefont {Massar}}, \bibinfo
  {author} {\bibfnamefont {A.}~\bibnamefont {de~la Giroday}}, \bibinfo {author}
  {\bibfnamefont {D.}~\bibnamefont {Matsukevich}}, \bibinfo {author}
  {\bibfnamefont {P.}~\bibnamefont {Maunz}}, \bibinfo {author} {\bibfnamefont
  {S.}~\bibnamefont {Olmschenk}}, \bibinfo {author} {\bibfnamefont
  {D.}~\bibnamefont {Hayes}}, \bibinfo {author} {\bibfnamefont
  {L.}~\bibnamefont {Luo}}, \bibinfo {author} {\bibfnamefont {T.}~\bibnamefont
  {Manning}}, \ and\ \bibinfo {author} {\bibfnamefont {C.}~\bibnamefont
  {Monroe}},\ }\href {http://www.ncbi.nlm.nih.gov/pubmed/20393558} {\bibfield
  {journal} {\bibinfo  {journal} {Nature}\ }\textbf {\bibinfo {volume} {464}},\
  \bibinfo {pages} {1021} (\bibinfo {year} {2010})}\BibitemShut {NoStop}%
\bibitem [{\citenamefont {Mayers}\ and\ \citenamefont
  {Yao}(2004)}]{selftestingmayers}%
  \BibitemOpen
  \bibfield  {author} {\bibinfo {author} {\bibfnamefont {D.}~\bibnamefont
  {Mayers}}\ and\ \bibinfo {author} {\bibfnamefont {A.~C.-C.}\ \bibnamefont
  {Yao}},\ }\href@noop {} {\bibfield  {journal} {\bibinfo  {journal} {Quantum
  Inf. Comput.}\ }\textbf {\bibinfo {volume} {4}},\ \bibinfo {pages} {273}
  (\bibinfo {year} {2004})}\BibitemShut {NoStop}%
\bibitem [{\citenamefont {Tsirelson}(1980)}]{Cirel'son1980}%
  \BibitemOpen
  \bibfield  {author} {\bibinfo {author} {\bibfnamefont {B.}~\bibnamefont
  {Tsirelson}},\ }\href
  {http://www.springerlink.com/index/l57053g573430450.pdf} {\bibfield
  {journal} {\bibinfo  {journal} {Lett. Math. Phys.}\ }\textbf {\bibinfo
  {volume} {4}},\ \bibinfo {pages} {93} (\bibinfo {year} {1980})}\BibitemShut
  {NoStop}%
\bibitem [{\citenamefont {Tsirelson}(1993)}]{Tsirelson1993}%
  \BibitemOpen
  \bibfield  {author} {\bibinfo {author} {\bibfnamefont {B.}~\bibnamefont
  {Tsirelson}},\ }\href@noop {} {\bibfield  {journal} {\bibinfo  {journal}
  {Hadronic J.}\ }\textbf {\bibinfo {volume} {8}},\ \bibinfo {pages} {329}
  (\bibinfo {year} {1993})},\ \bibinfo {note}
  {\url{http://www.tau.ac.il/~tsirel/download/hadron.html}}\BibitemShut
  {NoStop}%
\bibitem [{\citenamefont {Allcock}\ \emph {et~al.}(2009)\citenamefont
  {Allcock}, \citenamefont {Brunner}, \citenamefont {Pawlowski},\ and\
  \citenamefont {Scarani}}]{PhysRevA.80.040103}%
  \BibitemOpen
  \bibfield  {author} {\bibinfo {author} {\bibfnamefont {J.}~\bibnamefont
  {Allcock}}, \bibinfo {author} {\bibfnamefont {N.}~\bibnamefont {Brunner}},
  \bibinfo {author} {\bibfnamefont {M.}~\bibnamefont {Pawlowski}}, \ and\
  \bibinfo {author} {\bibfnamefont {V.}~\bibnamefont {Scarani}},\ }\href@noop
  {} {\bibfield  {journal} {\bibinfo  {journal} {Phys. Rev. A}\ }\textbf
  {\bibinfo {volume} {80}},\ \bibinfo {pages} {040103} (\bibinfo {year}
  {2009})}\BibitemShut {NoStop}%
\bibitem [{\citenamefont {Pawlowski}\ \emph {et~al.}(2009)\citenamefont
  {Pawlowski}, \citenamefont {Paterek}, \citenamefont {Kaszlikowski},
  \citenamefont {Scarani}, \citenamefont {Winter},\ and\ \citenamefont
  {Zukowski}}]{Pawlowski}%
  \BibitemOpen
  \bibfield  {author} {\bibinfo {author} {\bibfnamefont {M.}~\bibnamefont
  {Pawlowski}}, \bibinfo {author} {\bibfnamefont {T.}~\bibnamefont {Paterek}},
  \bibinfo {author} {\bibfnamefont {D.}~\bibnamefont {Kaszlikowski}}, \bibinfo
  {author} {\bibfnamefont {V.}~\bibnamefont {Scarani}}, \bibinfo {author}
  {\bibfnamefont {A.}~\bibnamefont {Winter}}, \ and\ \bibinfo {author}
  {\bibfnamefont {M.}~\bibnamefont {Zukowski}},\ }\href@noop {} {\bibfield
  {journal} {\bibinfo  {journal} {Nature}\ }\textbf {\bibinfo {volume} {461}},\
  \bibinfo {pages} {1} (\bibinfo {year} {2009})}\BibitemShut {NoStop}%
\bibitem [{\citenamefont {Gallego}\ \emph {et~al.}(2011)\citenamefont
  {Gallego}, \citenamefont {W\"urflinger}, \citenamefont {Ac\'\i{}n},\ and\
  \citenamefont {Navascu\'es}}]{PhysRevLett.107.210403}%
  \BibitemOpen
  \bibfield  {author} {\bibinfo {author} {\bibfnamefont {R.}~\bibnamefont
  {Gallego}}, \bibinfo {author} {\bibfnamefont {L.~E.}\ \bibnamefont
  {W\"urflinger}}, \bibinfo {author} {\bibfnamefont {A.}~\bibnamefont
  {Ac\'\i{}n}}, \ and\ \bibinfo {author} {\bibfnamefont {M.}~\bibnamefont
  {Navascu\'es}},\ }\href {\doibase 10.1103/PhysRevLett.107.210403} {\bibfield
  {journal} {\bibinfo  {journal} {Phys. Rev. Lett.}\ }\textbf {\bibinfo
  {volume} {107}},\ \bibinfo {pages} {210403} (\bibinfo {year}
  {2011})}\BibitemShut {NoStop}%
\bibitem [{\citenamefont {Navascu\'{e}s}\ and\ \citenamefont
  {Wunderlich}(2009)}]{Navascues2009}%
  \BibitemOpen
  \bibfield  {author} {\bibinfo {author} {\bibfnamefont {M.}~\bibnamefont
  {Navascu\'{e}s}}\ and\ \bibinfo {author} {\bibfnamefont {H.}~\bibnamefont
  {Wunderlich}},\ }\href
  {http://rspa.royalsocietypublishing.org/content/466/2115/881.short}
  {\bibfield  {journal} {\bibinfo  {journal} {Proc. Roy. Soc. Lond. A}\
  }\textbf {\bibinfo {volume} {466}},\ \bibinfo {pages} {881} (\bibinfo {year}
  {2009})}\BibitemShut {NoStop}%
\bibitem [{\citenamefont {Oppenheim}\ and\ \citenamefont
  {Wehner}(2010)}]{Oppenheim2010}%
  \BibitemOpen
  \bibfield  {author} {\bibinfo {author} {\bibfnamefont {J.}~\bibnamefont
  {Oppenheim}}\ and\ \bibinfo {author} {\bibfnamefont {S.}~\bibnamefont
  {Wehner}},\ }\href {http://www.ncbi.nlm.nih.gov/pubmed/21097930} {\bibfield
  {journal} {\bibinfo  {journal} {Science}\ }\textbf {\bibinfo {volume}
  {330}},\ \bibinfo {pages} {1072} (\bibinfo {year} {2010})}\BibitemShut
  {NoStop}%
\bibitem [{\citenamefont {Cabello}(2013)}]{Cabello2013}%
  \BibitemOpen
  \bibfield  {author} {\bibinfo {author} {\bibfnamefont {A.}~\bibnamefont
  {Cabello}},\ }\href@noop {} {\bibfield  {journal} {\bibinfo  {journal} {Phys.
  Rev. Lett.}\ }\textbf {\bibinfo {volume} {110}},\ \bibinfo {pages} {060402}
  (\bibinfo {year} {2013})}\BibitemShut {NoStop}%
\bibitem [{\citenamefont {Wehner}(2006)}]{Wehner2006}%
  \BibitemOpen
  \bibfield  {author} {\bibinfo {author} {\bibfnamefont {S.}~\bibnamefont
  {Wehner}},\ }\href {http://pra.aps.org/abstract/PRA/v73/i2/e022110}
  {\bibfield  {journal} {\bibinfo  {journal} {Phys. Rev. A}\ }\textbf {\bibinfo
  {volume} {73}},\ \bibinfo {pages} {1} (\bibinfo {year} {2006})}\BibitemShut
  {NoStop}%
\bibitem [{\citenamefont {Navascu\'es}\ \emph {et~al.}(2007)\citenamefont
  {Navascu\'es}, \citenamefont {Pironio},\ and\ \citenamefont
  {Ac{\'\i}n}}]{PhysRevLett.98.010401}%
  \BibitemOpen
  \bibfield  {author} {\bibinfo {author} {\bibfnamefont {M.}~\bibnamefont
  {Navascu\'es}}, \bibinfo {author} {\bibfnamefont {S.}~\bibnamefont
  {Pironio}}, \ and\ \bibinfo {author} {\bibfnamefont {A.}~\bibnamefont
  {Ac{\'\i}n}},\ }\href {\doibase 10.1103/PhysRevLett.98.010401} {\bibfield
  {journal} {\bibinfo  {journal} {Phys. Rev. Lett.}\ }\textbf {\bibinfo
  {volume} {98}},\ \bibinfo {pages} {010401} (\bibinfo {year}
  {2007})}\BibitemShut {NoStop}%
\bibitem [{\citenamefont {Navascu\'{e}s}\ \emph {et~al.}(2008)\citenamefont
  {Navascu\'{e}s}, \citenamefont {Pironio},\ and\ \citenamefont
  {Ac\'{\i}n}}]{Navascues2008}%
  \BibitemOpen
  \bibfield  {author} {\bibinfo {author} {\bibfnamefont {M.}~\bibnamefont
  {Navascu\'{e}s}}, \bibinfo {author} {\bibfnamefont {S.}~\bibnamefont
  {Pironio}}, \ and\ \bibinfo {author} {\bibfnamefont {A.}~\bibnamefont
  {Ac\'{\i}n}},\ }\href {http://iopscience.iop.org/1367-2630/10/7/073013}
  {\bibfield  {journal} {\bibinfo  {journal} {New J. Phys.}\ }\textbf {\bibinfo
  {volume} {10}},\ \bibinfo {pages} {1} (\bibinfo {year} {2008})}\BibitemShut
  {NoStop}%
\bibitem [{\citenamefont {{Doherty}}\ \emph {et~al.}(2008)\citenamefont
  {{Doherty}}, \citenamefont {{Liang}}, \citenamefont {{Toner}},\ and\
  \citenamefont {{Wehner}}}]{2008arXiv0803.4373D}%
  \BibitemOpen
  \bibfield  {author} {\bibinfo {author} {\bibfnamefont {A.~C.}\ \bibnamefont
  {{Doherty}}}, \bibinfo {author} {\bibfnamefont {Y.-C.}\ \bibnamefont
  {{Liang}}}, \bibinfo {author} {\bibfnamefont {B.}~\bibnamefont {{Toner}}}, \
  and\ \bibinfo {author} {\bibfnamefont {S.}~\bibnamefont {{Wehner}}},\
  }\href@noop {} {\bibfield  {journal} {\bibinfo  {journal} {P. IEEE CCC}\ ,\
  \bibinfo {pages} {199}} (\bibinfo {year} {2008})}\BibitemShut {NoStop}%
\bibitem [{\citenamefont {V\'{e}rtesi}\ and\ \citenamefont
  {P\'{a}l}(2009)}]{Vertesi2009}%
  \BibitemOpen
  \bibfield  {author} {\bibinfo {author} {\bibfnamefont {T.}~\bibnamefont
  {V\'{e}rtesi}}\ and\ \bibinfo {author} {\bibfnamefont {K.}~\bibnamefont
  {P\'{a}l}},\ }\href@noop {} {\bibfield  {journal} {\bibinfo  {journal} {Phys.
  Rev. A}\ }\textbf {\bibinfo {volume} {79}},\ \bibinfo {pages} {042106}
  (\bibinfo {year} {2009})}\BibitemShut {NoStop}%
\bibitem [{\citenamefont {Brunner}\ \emph {et~al.}(2008)\citenamefont
  {Brunner}, \citenamefont {Pironio}, \citenamefont {Acin}, \citenamefont
  {Gisin}, \citenamefont {M\'ethot},\ and\ \citenamefont
  {Scarani}}]{PhysRevLett.100.210503}%
  \BibitemOpen
  \bibfield  {author} {\bibinfo {author} {\bibfnamefont {N.}~\bibnamefont
  {Brunner}}, \bibinfo {author} {\bibfnamefont {S.}~\bibnamefont {Pironio}},
  \bibinfo {author} {\bibfnamefont {A.}~\bibnamefont {Acin}}, \bibinfo {author}
  {\bibfnamefont {N.}~\bibnamefont {Gisin}}, \bibinfo {author} {\bibfnamefont
  {A.~A.}\ \bibnamefont {M\'ethot}}, \ and\ \bibinfo {author} {\bibfnamefont
  {V.}~\bibnamefont {Scarani}},\ }\href@noop {} {\bibfield  {journal} {\bibinfo
   {journal} {Phys. Rev. Lett.}\ }\textbf {\bibinfo {volume} {100}},\ \bibinfo
  {pages} {210503} (\bibinfo {year} {2008})}\BibitemShut {NoStop}%
\bibitem [{\citenamefont {Gallego}\ \emph {et~al.}(2010)\citenamefont
  {Gallego}, \citenamefont {Brunner}, \citenamefont {Hadley},\ and\
  \citenamefont {Ac\'{\i}n}}]{PhysRevLett.105.230501}%
  \BibitemOpen
  \bibfield  {author} {\bibinfo {author} {\bibfnamefont {R.}~\bibnamefont
  {Gallego}}, \bibinfo {author} {\bibfnamefont {N.}~\bibnamefont {Brunner}},
  \bibinfo {author} {\bibfnamefont {C.}~\bibnamefont {Hadley}}, \ and\ \bibinfo
  {author} {\bibfnamefont {A.}~\bibnamefont {Ac\'{\i}n}},\ }\href@noop {}
  {\bibfield  {journal} {\bibinfo  {journal} {Phys. Rev. Lett.}\ }\textbf
  {\bibinfo {volume} {105}},\ \bibinfo {pages} {230501} (\bibinfo {year}
  {2010})}\BibitemShut {NoStop}%
\bibitem [{\citenamefont {Moroder}\ \emph {et~al.}(2013)\citenamefont
  {Moroder}, \citenamefont {Bancal}, \citenamefont {Liang}, \citenamefont
  {Hofmann},\ and\ \citenamefont {G\"uhne}}]{PhysRevLett.111.030501}%
  \BibitemOpen
  \bibfield  {author} {\bibinfo {author} {\bibfnamefont {T.}~\bibnamefont
  {Moroder}}, \bibinfo {author} {\bibfnamefont {J.-D.}\ \bibnamefont {Bancal}},
  \bibinfo {author} {\bibfnamefont {Y.-C.}\ \bibnamefont {Liang}}, \bibinfo
  {author} {\bibfnamefont {M.}~\bibnamefont {Hofmann}}, \ and\ \bibinfo
  {author} {\bibfnamefont {O.}~\bibnamefont {G\"uhne}},\ }\href {\doibase
  10.1103/PhysRevLett.111.030501} {\bibfield  {journal} {\bibinfo  {journal}
  {Phys. Rev. Lett.}\ }\textbf {\bibinfo {volume} {111}},\ \bibinfo {pages}
  {030501} (\bibinfo {year} {2013})}\BibitemShut {NoStop}%
\bibitem [{\citenamefont {Wolf}\ and\ \citenamefont
  {Perez-Garcia}(2009)}]{PhysRevLett.102.190504}%
  \BibitemOpen
  \bibfield  {author} {\bibinfo {author} {\bibfnamefont {M.~M.}\ \bibnamefont
  {Wolf}}\ and\ \bibinfo {author} {\bibfnamefont {D.}~\bibnamefont
  {Perez-Garcia}},\ }\href@noop {} {\bibfield  {journal} {\bibinfo  {journal}
  {Phys. Rev. Lett.}\ }\textbf {\bibinfo {volume} {102}},\ \bibinfo {pages}
  {190504} (\bibinfo {year} {2009})}\BibitemShut {NoStop}%
\bibitem [{\citenamefont {Wehner}\ \emph {et~al.}(2008)\citenamefont {Wehner},
  \citenamefont {Christandl},\ and\ \citenamefont
  {Doherty}}]{PhysRevA.78.062112}%
  \BibitemOpen
  \bibfield  {author} {\bibinfo {author} {\bibfnamefont {S.}~\bibnamefont
  {Wehner}}, \bibinfo {author} {\bibfnamefont {M.}~\bibnamefont {Christandl}},
  \ and\ \bibinfo {author} {\bibfnamefont {A.~C.}\ \bibnamefont {Doherty}},\
  }\href@noop {} {\bibfield  {journal} {\bibinfo  {journal} {Phys. Rev. A}\
  }\textbf {\bibinfo {volume} {78}},\ \bibinfo {pages} {062112} (\bibinfo
  {year} {2008})}\BibitemShut {NoStop}%
\bibitem [{\citenamefont {Pitowsky}(1986)}]{Pitowsky1986}%
  \BibitemOpen
  \bibfield  {author} {\bibinfo {author} {\bibfnamefont {I.}~\bibnamefont
  {Pitowsky}},\ }\href@noop {} {\bibfield  {journal} {\bibinfo  {journal} {J.
  Math. Phys.}\ }\textbf {\bibinfo {volume} {27}},\ \bibinfo {pages} {1556}
  (\bibinfo {year} {1986})}\BibitemShut {NoStop}%
\bibitem [{\citenamefont {Peres}(1999)}]{Peres1999}%
  \BibitemOpen
  \bibfield  {author} {\bibinfo {author} {\bibfnamefont {A.}~\bibnamefont
  {Peres}},\ }\href@noop {} {\bibfield  {journal} {\bibinfo  {journal} {Found.
  Phys.}\ }\textbf {\bibinfo {volume} {29}},\ \bibinfo {pages} {589} (\bibinfo
  {year} {1999})}\BibitemShut {NoStop}%
\bibitem [{\citenamefont {Avis}\ \emph {et~al.}(2004)\citenamefont {Avis},
  \citenamefont {Imai}, \citenamefont {Ito},\ and\ \citenamefont
  {Sasaki}}]{Avis2004}%
  \BibitemOpen
  \bibfield  {author} {\bibinfo {author} {\bibfnamefont {D.}~\bibnamefont
  {Avis}}, \bibinfo {author} {\bibfnamefont {H.}~\bibnamefont {Imai}}, \bibinfo
  {author} {\bibfnamefont {T.}~\bibnamefont {Ito}}, \ and\ \bibinfo {author}
  {\bibfnamefont {Y.}~\bibnamefont {Sasaki}},\ }\href
  {http://arxiv.org/abs/quant-ph/0404014} {\bibfield  {journal} {\bibinfo
  {journal} {pre-print}\ } (\bibinfo {year} {2004})},\ \Eprint
  {http://arxiv.org/abs/quant-ph/0404014v3} {arXiv:quant-ph/0404014v3}
  \BibitemShut {NoStop}%
\bibitem [{\citenamefont {Budroni}\ and\ \citenamefont
  {Cabello}(2012)}]{Budroni}%
  \BibitemOpen
  \bibfield  {author} {\bibinfo {author} {\bibfnamefont {C.}~\bibnamefont
  {Budroni}}\ and\ \bibinfo {author} {\bibfnamefont {A.}~\bibnamefont
  {Cabello}},\ }\href {http://iopscience.iop.org/1751-8121/45/38/385304}
  {\bibfield  {journal} {\bibinfo  {journal} {J. Phys. A-Math. Theor.}\
  }\textbf {\bibinfo {volume} {45}},\ \bibinfo {pages} {385304} (\bibinfo
  {year} {2012})}\BibitemShut {NoStop}%
\bibitem [{\citenamefont {Heydari}(2006)}]{Heydari2006}%
  \BibitemOpen
  \bibfield  {author} {\bibinfo {author} {\bibfnamefont {H.}~\bibnamefont
  {Heydari}},\ }\href {http://stacks.iop.org/0305-4470/39/i=38/a=012}
  {\bibfield  {journal} {\bibinfo  {journal} {J. Phys. A-Math. Gen.}\ }\textbf
  {\bibinfo {volume} {39}},\ \bibinfo {pages} {1} (\bibinfo {year}
  {2006})}\BibitemShut {NoStop}%
\bibitem [{\citenamefont {Fishburn}\ and\ \citenamefont
  {Reeds}(1994)}]{Fishburn1994}%
  \BibitemOpen
  \bibfield  {author} {\bibinfo {author} {\bibfnamefont {P.}~\bibnamefont
  {Fishburn}}\ and\ \bibinfo {author} {\bibfnamefont {J.}~\bibnamefont
  {Reeds}},\ }\href@noop {} {\bibfield  {journal} {\bibinfo  {journal} {SIAM J.
  Discrete Math.}\ }\textbf {\bibinfo {volume} {7}},\ \bibinfo {pages} {48}
  (\bibinfo {year} {1994})}\BibitemShut {NoStop}%
\bibitem [{\citenamefont {Golub}\ and\ \citenamefont
  {Van~Loan}(2013)}]{MatrixComputations}%
  \BibitemOpen
  \bibfield  {author} {\bibinfo {author} {\bibfnamefont {G.~H.}\ \bibnamefont
  {Golub}}\ and\ \bibinfo {author} {\bibfnamefont {C.~F.}\ \bibnamefont
  {Van~Loan}},\ }\href@noop {} {\emph {\bibinfo {title} {Matrix
  Computations}}},\ \bibinfo {edition} {4th}\ ed.\ (\bibinfo  {publisher} {The
  Johns Hopkins University Press},\ \bibinfo {address} {Baltimore},\ \bibinfo
  {year} {2013})\BibitemShut {NoStop}%
\bibitem [{\citenamefont {Cabello}\ \emph {et~al.}(2009)\citenamefont
  {Cabello}, \citenamefont {Larsson},\ and\ \citenamefont
  {Rodr\'\i{}guez}}]{Cabello2009}%
  \BibitemOpen
  \bibfield  {author} {\bibinfo {author} {\bibfnamefont {A.}~\bibnamefont
  {Cabello}}, \bibinfo {author} {\bibfnamefont {J.-A.}\ \bibnamefont
  {Larsson}}, \ and\ \bibinfo {author} {\bibfnamefont {D.}~\bibnamefont
  {Rodr\'\i{}guez}},\ }\href@noop {} {\bibfield  {journal} {\bibinfo  {journal}
  {Phys. Rev. A}\ }\textbf {\bibinfo {volume} {79}} (\bibinfo {year}
  {2009})}\BibitemShut {NoStop}%
\end{thebibliography}
